\documentclass[10pt, conference, letterpaper]{IEEEtran}

\newcommand{\subparagraph}{}
\usepackage[compact]{titlesec}
\titlespacing{\section}{0pt}{1ex}{0ex}
\titlespacing{\subsection}{0pt}{0ex}{0ex}
\titlespacing{\subsubsection}{0pt}{0ex}{0ex}

\usepackage{booktabs}   
\usepackage{subcaption} 
\usepackage{latexsym}
\usepackage{amsfonts}
\usepackage{amstext}
\usepackage{array}
\usepackage{amsmath}
\usepackage{amssymb}
\usepackage{amsthm}
\usepackage{hyperref}
\usepackage{float}
\usepackage{caption}
\usepackage{framed}
\usepackage{listings}
\usepackage{color}
\usepackage{blkarray}
\usepackage[normalem]{ulem}

\usepackage{verbatim}
\usepackage{graphicx}
\usepackage[noend]{algpseudocode}

\usepackage{verbatim}
\usepackage{graphicx}

\newcounter{algorithm}
\algrenewcommand\textproc{\textsc}

\usepackage{cite}
\usepackage{xcolor}
\IEEEoverridecommandlockouts

\long\def\ruzica#1{{\color{red}{\bf Ruzica: }{\small #1}}}
\long\def\timos#1{{\color{blue}{\bf Timos: }{\small #1}}}

\long\def\qiao#1{{\color{orange}{\bf Qiao: }{\small #1}}}

\newtheoremstyle{exampstyle}
  {0em} 
  {0em} 
  {} 
  {} 
  {\bfseries} 
  {.} 
  {0em} 
  {} 
\theoremstyle{exampstyle}

\usepackage[linesnumbered,ruled,vlined]{algorithm2e}

\definecolor{Enc}{gray}{0.6}

\newcommand{\para}[1]{\noindent {\bf #1}}
\newcommand{\ie}{{\em i.e.}}
\newcommand{\eg}{{\em e.g.}}

\newcommand{\mat}[1]{\mathbf{#1}}
\newcommand{\nats}{\mathbb{N}}

\newcommand{\system}{IVeri}
\newtheorem{problem}{\textbf{Problem}}
\newtheorem{thm}{\textbf{Theorem}}

\newcommand{\A}{\mathbf{A}}
\newcommand{\B}{\mathbf{B}}

\newcommand{\EC}[3]{#1 \oplus \mathbf{PRF}(#2, #3)}
\newcommand{\FS}{\pi(\F)}
\newcommand{\OC}{O}

\newcommand{\as}{\mathbf{assign}}
\newcommand{\asi}{assign}
\newcommand{\pr}{\mathbf{prior}}

\newcommand{\F}{\mathbf{F}}
\newcommand{\Enc}[3]{#1\oplus\mathbf{PRF}(#2, #3)}

\usepackage[inline]{enumitem}
\begin{document}
\title{\system{}: Privacy-Preserving Interdomain Verification}

\author{
	\textit{Technical Report}\\
	Ning Luo$^{\ddagger}$,
		Qiao Xiang$^{\dagger\ddagger}$\thanks{Ning Luo and Qiao Xiang are co-primary authors. The corresponding author is Qiao Xiang (qiaoxiang@xmu.edu.cn).},
	Timos Antonopoulos$^{\ddagger}$,
	Ruzica Piskac$^{\ddagger}$,
	Y. Richard Yang$^{\ddagger}$,
	Franck Le$^{\ast}$
 \\	
	$^\dagger$Xiamen University,
	$^{\ddagger}$Yale University,
	$^{\ast}$IBM Watson Research Center
}

%

\maketitle

\begin{abstract}
In an interdomain network, autonomous systems (ASes) often establish peering
	agreements, so that one AS (agreement consumer) can influence the
	routing policies of the other AS (agreement provider). Peering
	agreements are implemented in the BGP configuration of the agreement
	provider. It is crucial to verify their implementation because one error
	can lead to disastrous consequences. However, the fundamental challenge
	for peering agreement verification is how to preserve the
	\textit{privacy} of both ASes involved in the agreement. To this end,
	this paper presents \system{}, the first privacy-preserving interdomain
	agreement verification system. \system{} models the interdomain
	agreement verification problem as a SAT formula, and develops a novel,
	efficient, privacy-serving SAT solver, which uses oblivious shuffling
	and garbled circuits as the key building blocks to let the agreement
	consumer and provider collaboratively verify the implementation of
	interdomain peering agreements without exposing their private
	information. A prototype of \system{} is implemented and evaluated
	extensively. Results show that \system{} achieves accurate,
	privacy-preserving interdomain agreement verification with reasonable
	overhead.
\end{abstract}

\section{Introduction}
\label{sec:intro}

An interdomain network (\eg, the Internet~\cite{gao2001stable} and the Large
Hadron Collider science network~\cite{lhc}) connects multiple autonomous systems
(ASes) and exchanges traffic among them. The \textit{de facto} interdomain
routing protocol is the Border Gateway Protocol~\cite{BGP}. ASes use BGP to
exchange routing information. In addition, BGP also allows each AS to
independently apply its local routing policies to select and advertise
interdomain routes. 


Due to economic or collaborative incentives, two peering ASes often establish
\textit{peering agreements} so that one AS (\textit{agreement consumer}) can
influence the routing policies of the other AS (\textit{agreement provider}).
Such agreements are implemented in the configurations of the provider's BGP
routers. Consider an interdomain network in Figure~\ref{fig:motivation-example}.
An example of a peering agreement between ASes $A$ and $B$ is that $B$ will send
all $A$'s traffic towards $F$ through $D$. In this example, $A$ is the agreement
consumer, and $B$ is the agreement provider, who implements this agreement in its
BGP configurations. 



The correct implementation of peering agreements is crucial for ASes to realize
their network management goals, \eg, defense against distributed
denial-of-service (DDoS) attacks, inbound and outbound traffic engineering, and
fulfillment of regulations on data traffic traversal. Nevertheless, technology
news have been repeatedly reporting about peering agreements not being
implemented correctly, leading to disastrous results~\cite{timewarner,
bgpleak-japan, bgpleak-aws, mahajan2002understanding}. For example,
due to an error in a BGP configuration in 2018, Google lost control of several
millions of IP addresses for more than one hour and rerouted the internet
traffic through China~\cite{bgp-google}. 


\begin{figure}[t]
\centering
\includegraphics[width=0.7\linewidth]{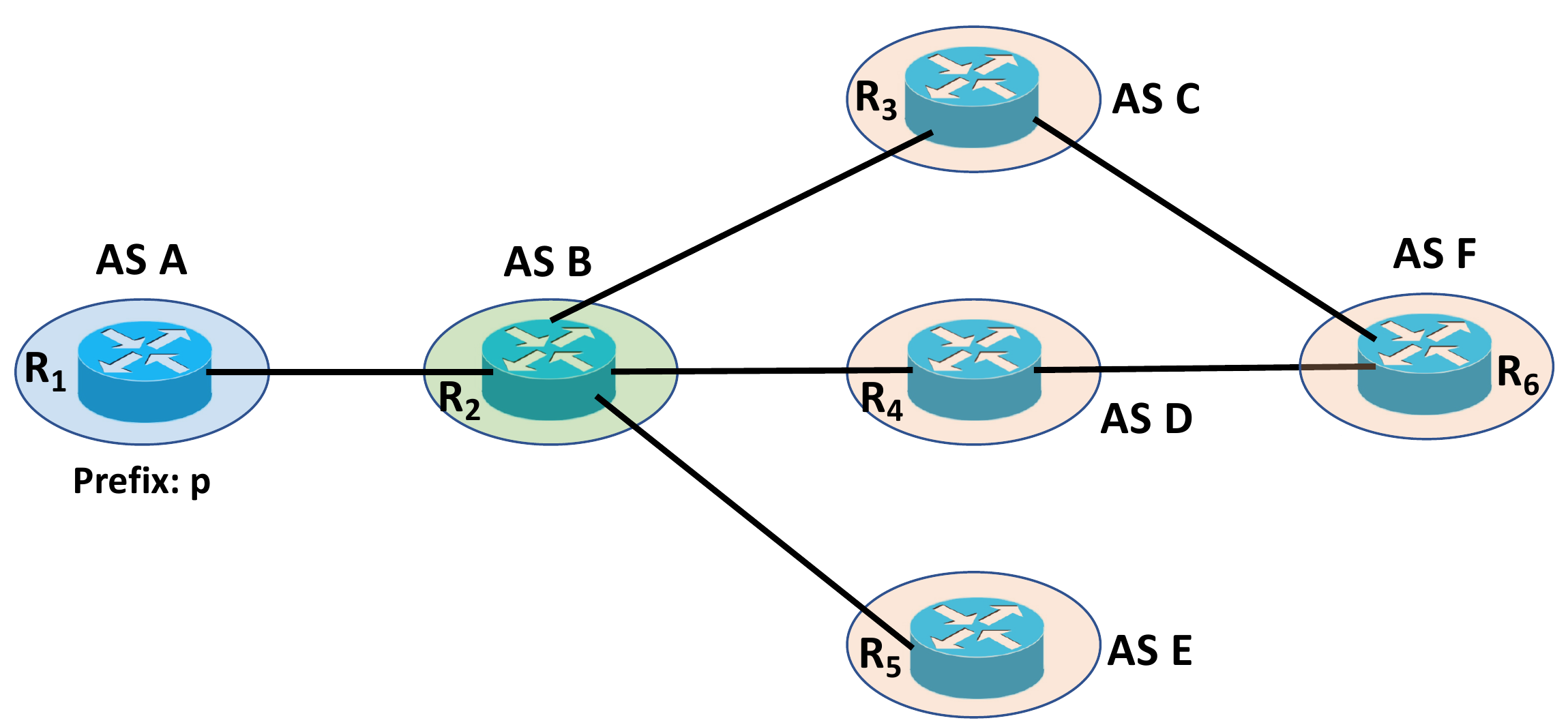}
\caption{An example of interdomain network to illustrate BGP and interdomain
	peering agreements.}
\label{fig:motivation-example}
\end{figure}


As such, it is crucial to verify whether a peering agreement is implemented
correctly at the provider's BGP configurations proactively, before the buggy
configurations are deployed or before the agreement is actual invoked (\ie, via
a route announcement with community tag). However, this problem remains
unsolved and non-trivial, despite the fact that many recent tools are focused on finding
configuration errors in routing protocols~\cite{minesweeper, compression, era, fsr, bagpipe}. 
The fundamental challenge for this problem is
how to preserve the \textit{privacy} of both ASes involved in the agreement. On
one hand, the BGP configurations of the agreement provider is private for
operational security~\cite{anwar2015investigating}. On the other hand, our
private conversations with some of the largest content service providers, who
play the role of agreement consumer, show that if possible, they would also
prefer their agreements-to-verify being kept private, for the reason that
exposing these may cause their sensitive information (\eg, internal operational
policies and offline negotiation with other ASes) being inferred. With the
privacy concerns of both parties, strawman solutions such as the client-server design and
the trusted third-party design to apply recent network configuration tools are
not desired. A related approach~\cite{vperf, netreview, sam-icnp-collaborative} is to
design collaborative verification protocols in which an agreement consumer can
verify how the agreement provider processes an actual route announcement 
with the help of all other neighboring ASes of the provider. However, this
approach verifies peering agreement on per route announcement basis, \ie, it can
only find errors after a route announcement is sent.


In this paper, we systematically study this problem and design \system{}, the
first privacy-preserving interdomain agreement verification system.  In
particular, we first leverage the formalism from recent network configuration
verification tools~\cite{minesweeper, compression, era} to let the agreement consumer
model the agreement-to-verify and the agreement provider model the BGP
configurations as SMT (satisfiability modulo theories) formulas and transform to
SAT formulas, respectively. 
Specifically, let $F^B$
represent the BGP configurations of the agreement provider and $F^C$ represent
the SAT formula describing the agreement-to-verify of the
agreement consumer. 
To verify that provider's configurations correctly implement the
agreement-to-verify is to check whether $F^B \Rightarrow F^C$ is a tautology,
which is equivalent to check whether $F^B \land \lnot F^C$ is
unsatisfiable~\cite{minesweeper, compression, era}. To simplify notation, we
denote $\lnot F^C$ as 
$F^A$ and thus the problem of privacy-preserving interdomain agreement
verification is defined as to \textit{determine the satisfiability of $F^A \wedge
F^B$, while $F^A$ and $F^B$ are held privately by the agreement consumer and
provider, respectively}.

To this end, we developed a novel, efficient,
privacy-preserving SAT solver in \system{}. The basic idea of this solver consists of (1) an
oblivious shuffling algorithm \cite{oblivioushuffling} that enables two parties
to shuffle an encrypted version of the verification formula, and further ensures
that the non-encrypted configuration and agreement formulas are always held
private by the respective parties, and (2) the standard DPLL algorithm encoded
with inexpensive garbled circuits \cite{Garbledcircuits}, denoted by GC-DPLL.
The solver determines the satisfiability of the verification formula contributed
by the agreement consumer and provider without loss of their privacy on their
contribution.  


In summary, this paper presents the following \textbf{main contributions}:
\begin{itemize}[leftmargin=*]
	\item We systematically study an important, real problem of
		interdomain peering agreement verification, and present
		\system{}, a privacy-preserving interdomain agreement
		verification system, which is, to the best of our knowledge, the
		first system to verify the agreement implementation in BGP
		configurations;
		
	\item \system{} leverages and expands the modeling technique in recent
		network verification tools and establishes the interdomain peering
		agreement verification problem as a privacy-preserving SAT problem;
	\item \system{} designs a novel, efficient privacy-preserving SAT
		solver that determines the satisfiability of the verification
		formula without leaking the agreement consumer or provider sensitive information. In addition to interdomain verification, this SAT solver also has other application prospectives such as multi-domain resource orchestration via constraint programming~\cite{xiang2019toward}; 
		
	\item A prototype of \system{} is implemented and evaluated with extensive experiments. 
		Results show that \system{} achieves accurate, privacy-preserving interdomain agreement verification with reasonable overhead. 
\end{itemize}

\section{Background, Motivation and Challenge}
\label{sec:background}

This section provides a brief background on BGP and interdomain peering
agreements, demonstrates the importance of interdomain peering agreement
verification, and elaborates on its fundamental challenges.

\subsection{Background}\label{sec:bgp-background}
\para{BGP in a nutshell}. 
BGP~\cite{BGP} is the \textit{de facto} interdomain routing protocol
interconnecting ASes. Specifically, each AS owns a set of IP addresses, and
assigns some of its routers as BGP border routers, which are connected to BGP
border routers in neighboring ASes.  
Figure~\ref{fig:motivation-example} shows an interdomain network with 6 ASes,
where AS $A$ owns IP prefix $p$. For simplicity of the presentation, we assume that each of them has one BGP
border router. 

Two ASes whose BGP routers are connected to each other are called BGP peers,
\eg, $A$ and $B$ are BGP peers.  BGP peers exchange routing information on how
to reach destination IP prefixes through \textit{route announcements}.
Abstractly, each route announcement carries a destination prefix and a sequence
of ASes to traverse to reach the destination prefix, which we call an \textit{AS
path}. In addition, a route announcement can also
carry other attributes, such as \textit{origin}, which is the origin AS of the
destination prefix, or \textit{community tags}, which are numerical values whose
semantics are predefined between two BGP peers via offline communication. For
example, $R_2$ can send a route announcement $(p, [B, A], comtag:10:30)$ to
$R_3$, which includes an destination prefix $p$, an AS path $[B, A]$ to reach $p$, 
and a community tag attribute of value $10:30$.

In BGP, each AS can make and execute its own policies to  decide which routes to
use (\ie, selection policy), and whether or not to announce the to peers (\ie,
export policy).  For example, an AS may prefer to selecting shorter routes for
better latency, and may choose not to announce a route to a peer for business
reason. In reality, these policies are implemented in the configurations of BGP
routers.

\para{BGP peering agreements}.
For economical or collaborative reasons, two BGP peering ASes often reach
peering agreements to allow one AS (agreement consumer) to influence the routing
policies of the other AS (agreement provider)~\cite{com-guide}.  Such agreements
are implemented as part of the route selection/export polices in the BGP
configurations of the agreement provider. Their correct implementation is
essential for ASes to realize their network management goals, \eg, defense
against DoS attacks, inbound and outbound traffic engineering, and fulfillment
of regulations on data traffic traversal. Some representative peering
agreements~\cite{com-guide, vperf} are:

\begin{itemize}[leftmargin=*]
	\item \textbf{Selective export}: let the agreement provider not 
		propagate the route announced by the
		consumer to certain peers of the providerr;  
	\item \textbf{Set local preference}: let the agreement provider set its
		local preference of the route announced by the consumer to
		certain value;
	\item \textbf{Prefer/avoid certain AS}: let the agreement provider
		prefer to /avoid selecting a route containing certain AS when receiving multiple
		route announcement;
\end{itemize}

\subsection{Motivation}\label{sec:motivation}
Verifying whether an interdomain peering agreement is correctly implemented in
BGP configurations is of great importance because one error can lead to
harmful consequences~\cite{timewarner,
bgpleak-japan, bgpleak-aws, mahajan2002understanding}. We use the following example to demonstrate
that.

\para{A motivating example (DoS attack)}.
Consider the interdomain network given in Figure~\ref{fig:motivation-example},
and assume that AS $E$ launches a DoS attack to $A$, \ie, keeps sending high
volumes of malicious traffic to
prefix $p$ in $A$ within a short time period. When $A$ detects a DoS attack,
and suspects that $E$ is the source of this attack, $A$ can defend 
by reaching a \textit{selective-export} agreement with $B$: \textit{$B$ will not
export any route announced by $A$ to $E$}. If such an agreement is implemented
correctly by $B$, $B$ will send a new route announcement $(p, [])$ to $E$, also
called a withdraw announcement, indicating that $B$ cannot reach $p$. As a
result, $E$ cannot continue the attack because it has no route to reach $p$, and
$A$ can observe the stop of the attack.

However, if $B$ does not implement this agreement correctly, $E$ will
continue having a route $B \rightarrow A$ to reach prefix $p$, and use this
route for the attack. As such, $A$ will observe that the attack does not stop,
and shift its suspicion to another AS, \eg, $C$, and reaches another
selective-export agreement with $B$ to not export any route announced by $A$ to
$C$. If this new agreement is implemented by $B$ correctly this time, not only
does $A$ still suffer from the DoS attack, the normal traffic from $C$ to $A$
is also completely blocked, exacerbating the damage of this attack.

\subsection{Challenges}
Many tools have been developed to find configurations errors in routing
protocols~\cite{minesweeper, compression, ms-local, era, fsr,
bagpipe}. As such, a strawman solution to verify whether an peering agreement is
correctly implemented is to treat the agreement the customer wants to verify as
a network property, let a trusted third party collect the agreement from the
consumer and the configurations from the provider and run an existing
configuration verification tool to verify it. However, this is not the case.

%

The reason, and the fundamental challenge of verifying interdomain peering
agreements, is \textit{privacy}. On one hand, given an agreement provider, its
implementation of peering agreements (\ie, the BGP router configuration files)
is private for operational security and commercial
reasons~\cite{anwar2015investigating}. Although some work has shown it is
possible to infer some aspects of ASes' BGP configurations~\cite{gao2001inferring}, the
inference results have limited accuracy and cannot be utilized for peering
agreement verification.

On the other hand, for an agreement consumer, the peering agreements are also considered private. One may find this argument a little odd
at a first glance, however, it is justified. Many
agreements (\eg, selective export, set local preference and AS path prepending)
are invoked ``on-demand" to achieve certain network management goals (\eg,
inbound traffic engineering). When to invoke which agreements is related to the
operation policy of agreement consumers. As such, revealing which agreements to
verify risk exposing its operation policy. For example, consider the not-export-to-$E$
agreement between $A$ and $B$. Assume additionally that it is on-demand requiring $A$ attaching
community tags in its route announcements to $B$. If $A$ reveals to $B$
that it wants to verify that $B$ implements this agreement
correctly, but does not attach any community tags in its route announcements,
$B$ may deduce that either $A$ does not suspect $E$ as a DoS attack, or $A$
and $E$ reach additional commercial collaboration that allows $E$ to send traffic to
$A$, both of which are private to $A$.

\begin{figure}[t]
\centering
\includegraphics[width=0.65\linewidth]{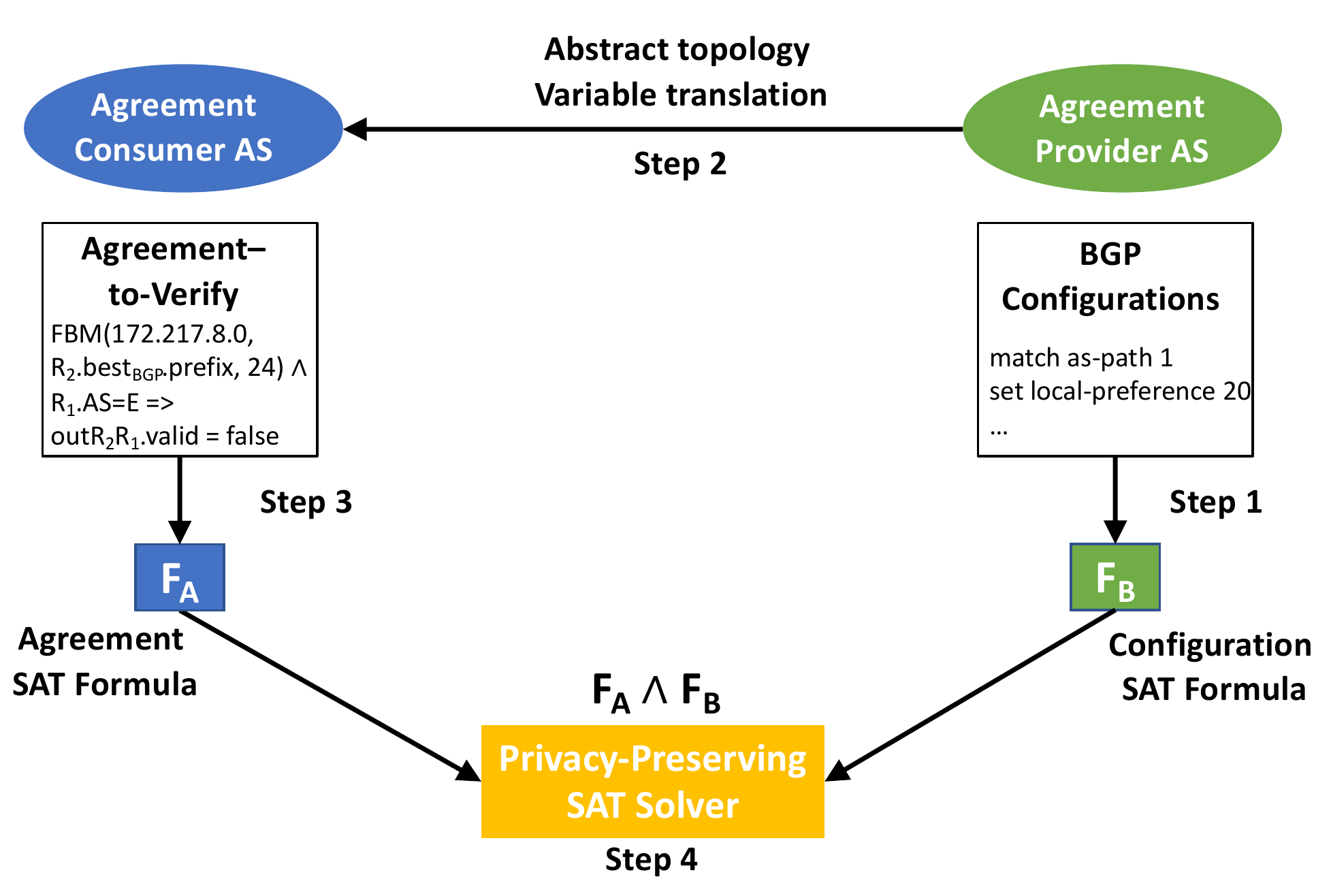}
\caption{The basic architecture and workflow of \system{}.}
\label{fig:arch}
\vspace{-0.5em}
\end{figure}

\section{System Model and Problem Formulation}\label{sec:settings}
\system{} operates between the agreement consumer and the agreement
provider~(Figure~\ref{fig:arch}). 
This section presents the system model of
\system{} and the formulation of the privacy-preserving interdomain agreement
verification problem.

\subsection{System Model}
We leverage the modeling technique of Minesweeper~\cite{minesweeper}, a
state-of-the-art intradomain network verification tool, to model the BGP
configurations of the agreement provider, but go beyond their focus on
intradomain network properties (\eg, reachability and black-hole freeness) and
safety properties of interdomain network (\eg, stability), to develop new models
for interdomain peering agreements.

\para{Agreement provider: encode BGP configurations as SAT formulas}. 
Similar as Minesweeper, in \system{}, the agreement provider extracts its BGP
configurations into an SMT formula. Due to space limit, we refer readers to
~\cite{minesweeper} for more details of this model. As an example, the following
SMT formula  specifies that a BGP router will not announce any route toward the
prefix $172.217.8.0/24$  to the BGP router whose IP address is $65.124.208.93$. 
\begin{equation}
	\footnotesize
	\begin{aligned}
		 FBM(172.217.8.0, best_{BGP}.prefix, 24) 
		\sf \wedge R_i.IP=65.124.208.93 \\
		\sf \Rightarrow  out_{R_4 R_i}.valid = false.
	\end{aligned}
\end{equation}
Because all the variables of the provider's SMT formula are finite-bounded, it
can be easily translated to a SAT formula in CNF (\ie, using bit blasting and Tseytin
transformation). We call the derived formula the \emph{configuration SAT formula}.

\para{Agreement provider: provide abstract BGP router topology and SMT/SAT
variables}. 
The agreement provider exposes to the consumer an abstract topology of a
full-mesh of the provider's BGP routers and their connections to other BGP
peers, and a SMT/SAT variable translation list. The abstract topology provides
sufficient topology information (\ie, all BGP connections of the agreement
provider) for the consumer to specify peering agreements to verify. The exposed
list of SMT/SAT variable translation ensures that the variable semantics are
consistent between the configuration SAT formula and the agreement SAT formula.
It also  allows the agreement provider to decide the scope of peering agreements
that the agreement consumer can verify, so as to restrict the consumer from
snooping-around with arbitrary verification requests.

\para{Agreement consumer: encode peering agreements as SAT formulas}.
With the abstract topology and the SMT/SAT variable list, the agreement consumer
expresses the agreement-to-verify as an SMT formula. 
Consider the selective-export agreement from the motivating example given in
Section~\ref{sec:motivation}, it can be modeled as the
following SMT formula: 
\begin{equation}
	\footnotesize
	\begin{aligned}
\sf\forall i=3, 4, 5.\ &\sf FBM(172.217.8.0, R_2.best_{BGP}.prefix, 24)  
		\sf \wedge R_i.AS=E \\
		&\Rightarrow  out_{R_2 R_i}.valid = false,
	\end{aligned}
\end{equation}
where prefix $p$ is $172.217.8.0/24$. Similarly, other common peering agreements can also
be expressed as SMT formulas. However, recall that we need to check the  unsatisfiability of the configuration SAT formula and the negation of the peering agreement. 
For the same reasons as when encoding BGP
configurations, the negation of the consumer's SMT formula can also be transformed to a
SAT formula in CNF. The resulting formula we call the \textit{agreement SAT
formula}.

\subsection{Problem Formulation}
After introducing the system model of \system{}, we next formally define the
privacy-preserving interdomain agreement verification problem.

\para{Notations}. 
The following conventions are followed to differentiate scalar, vector and
matrix variables. We use $v$ to represent a scalar and  $\mathbf{v}= \{v_1, v_2, \ldots,
v_{|\mathbf{v}|}\}$ to represent a row vector of size $|\mathbf{v}|$.
We use $\mathbf{V} = \{\mathbf{v}_1, \mathbf{v}_2, \ldots, \mathbf{v}_{|\mathbf{V}|}\}$, with $|\mathbf{v}_i|=|\mathbf{v}_j|$, for all $i\leq j\leq |\mathbf{V}|$
to represent a matrix of $|\mathbf{V}|$ rows. The element on the i-th row and $j$-th column of $\mathbf{V}$ is denoted by $\mathbf{V}_{ij}$. Unless explicitly
specified, operators on vectors or matrices indicate component-wise operations. 

A permutation $\pi$ is a bijective function: $\{1, \ldots, n\}\rightarrow \{1,
\ldots, n\}$, for $n\in\nats$. A \textit{row permutation} of a matrix
$\mathbf{V}$ is represented as $\{\mathbf{v}_{\pi^{-1}(1)}, \mathbf{v}_{\pi^{-1}(2)},
\ldots, \mathbf{v}_{\pi^{-1}(|\mathbf{V}|)}\}$ and written as $\pi(\mathbf{V})$.

A SAT formula in CNF of $m$ clauses (\ie, $c_1 \wedge c_2 \wedge ,\ldots, \wedge c_m$) over $n$ Boolean variables $x_1, x_2, \ldots, x_n$, is represented by an $n$-by-$m$ matrix $F$, where each element $\mathbf{F}_{ij}$ of $~\mathbf{F}$, consists of two bits. The first bit, denoted
as $\mathbf{F}_{ij}.O$, is 1 if variable $x_i$ occurs in clause $c_j$, and is 0
otherwise.
The second bit, denoted as $\mathbf{F}_{ij}.P$, is 1 if and only if  the variable $x_i$  occurs as a positive literal (\ie, $x_i$) in clause $c_j$.

Given a pair of peering ASes, $A$ and $B$ denote the agreement consumer and the
agreement provider, respectively.  As stated in Section~\ref{sec:intro}, the
agreement SAT formula (\ie,  the negation of the agreement-to-verify) generated
by $A$ is represented as $\mathbf{F}^A$, and the configuration SAT formula
generated by the provider $B$ is represented as $\mathbf{F}^B$. The
privacy-preserving interdomain peering agreement verification problem is
defined as follows: 

\begin{problem}[Privacy-Preserving Interdomain Peering Agreement Verification
	Problem]\label{prob:ppsat}
	Given an agreement consumer $A$ with associated agreement SAT formula $\mathbf{F}^A$ and an agreement provider $B$ with an associated configuration SAT
	formula $\mathbf{F}^B$, the two parties $A$ and $B$ should decide if the formula $\mathbf{F}^A \wedge \mathbf{F}^B$ is
	satisfiable in a way such that $A$ (respectively $B$) does not know what $\mathbf{F}^B$ (respectively $\mathbf{F}^A$)  is. The agreement provider $B$
	satisfies the agreement requirements iff formula 
	$\mathbf{F}^A \wedge \mathbf{F}^B$ is unsatisfiable.
\end{problem}

\para{Security model}. 
This paper assumes a \textit{semi-honest} security model \cite{Garbledcircuits},
\ie, the agreement customer and the agreement provider will not deviate from the
workflow specified in \system{}, but merely try to gather information during its
execution~\cite{raykova2012secure}.  This is sufficient for multiple scenarios
of interdomain routing, including commercial Internet~\cite{netreview},
collaboration science networks where ASes collaboratively conduct common tasks
such as data transfers~\cite{lhc}, and military coalition
networks~\cite{mishra2017comparing}.

\section{Privacy-Preserving SAT Solver}\label{sec:algorithm}
This section presents details on our instantiation of privacy-preserving SAT
solver in \system{} to verify interdomain peering agreements in a
privacy-preserving manner. We first provide a short review on the
theoretical foundations this solver is built upon, and then present the basic
idea and details of our solver.
We also conduct rigorous analysis on the
correctness, efficiency and privacy-preservingness of the solver.

\subsection{Preliminaries}

\para{DPLL algorithm}. \system{}'s
privacy-preserving SAT solver is based on the classic DPLL algorithm, as many
modern SAT solvers do~\cite{malik2009boolean, gomes2008satisfiability, sorensson2005minisat, z3}. In essence, DPLL is a backtracking search algorithm. 
Given a SAT formula in CNF represented by a matrix $\mathbf{F}$, the algorithm
runs by choosing a Boolean variable $x_i$, assigning a truth value to it,
simplifying the formula by removing all clauses that are satisfied by the truth
value assignment of $x_i$ as well as removing all literals involving the
variable $x_i$ that are false under this assignment, and then recursively
checking if the simplified formula is satisfiable. If so, $\mathbf{F}$ is
satisfiable. Otherwise, the algorithm uses the same recursive check by assigning
the opposite truth value to $x_i$. To improve the efficiency of this
backtracking search process, many optimizations are introduced
~\cite{malik2009boolean, gomes2008satisfiability, sorensson2005minisat}. One key
optimization that most modern SAT solvers adopt into the original DPLL
 is \textit{unit literal search
and resolution}, which attempts to find a unit clause that contains only one
unassigned variable, and assigns a truth value to this variable to satisfy this
clause, before attempting anything else. This optimization substantially reduces
the search space, leading to significant improvement in search efficiency.

\para{Secure multi-party computation (SMPC)}. SMPC refers to a set of frameworks for multiple parties to jointly compute a function over
their inputs while keeping those inputs private~\cite{Garbledcircuits}. Although there is a
rich set of SMPC related techniques, this paper focuses on the following two that are adopted in the proposed privacy-preserving SAT solver.


\begin{itemize}[leftmargin=*]
\item \textbf{Garbled circuits}. This is a cryptographic protocol, introduced in~\cite{yao1986generate}, that allows two mistrusting
parties to jointly evaluate any function that can be expressed as a Boolean
circuit over their private inputs. An example application of this protocol is
		the Millionaires' Problem, where the goal is for two parties to
		decide who has more money without exposing how much money they
		each have~\cite{Garbledcircuits}. The output of a garbled
		circuit can be revealed to both parties, only one specific
		party, or no party but be taken as inputs of other garbled
		circuits. For more details, readers may refer
		to~\cite{detailGC}.
\item \textbf{Oblivious shuffling}. This technique \cite{oblivioushuffling}
	enables two parties to collaboratively decide a permutation $\pi$ over a vector of cipher texts  $\mathbf{\Phi} = \{\phi_1, \phi_{2}, \ldots, \phi_{|\mathbf{v}|}\})$ such that $\phi_i$ is the cipher text of some $v_i$. The protocol provides that (1) the output vector of this protocol is the vector $ \{\phi_{\pi(1)}, \phi_{\pi(2)}, \ldots, \phi_{\pi(|\mathbf{v}|)}\}$ where $\phi_{\pi(i)} $ is the cipher texts of $v_{\pi(i)}$; (2) the permutation $\pi$ is unknown to both parties unless they collaborate; (3) the plain text $v_1, \cdots, v_n$ is never revealed to any party during the protocol. 

\end{itemize}

\subsection{Privacy-Preserving SAT Solver: Details}\label{subsec: ppsat}
This subsection first describes the basic idea behind the privacy-preserving SAT
solver, followed by the technical details of its key components.

\para{Limitations of garbled circuits}. 
A strawman solution to Problem~\ref{prob:ppsat} is to construct a garbled circuit of the entire DPLL algorithm, because it can be expressed as a Boolean circuit. However, the resulting circuit can blow up, especially for the algorithms such as DPLL of which the circuit size is exponential to the size of inputs.  
Such an explosion of complexity can be overcome by revealing a part of the information,  exposing the history of backtracking-search, or leaving the access pattern of private inputs of $\A$ and $\B$ outside the garbled circuit. However, the privacy leak of this solution can be substantial and lead to the complete exposure of $\mathbf{F}^\A$ and $\mathbf{F}^\B$. 



	


\para{Basic idea}: 
The proposed SAT solver addresses this tradeoff between the high overhead of
accessing private inputs in the garbled circuit and the severe privacy leakage
of accessing private inputs outside the garbled circuit with a simple, yet
elegant idea: instead of hiding the search history with high overhead, expose a disguised search history that no party can recover independently.

Specifically, the solver consists of two key steps: \textit{permute} and
\textit{garble}. 
First, two parties construct and permute a matrix corresponding to $\F^\A \land \F^\B$ through
the oblivious shuffling algorithm~(Algorithm~\ref{alg:os}), after which both permutation and entries of the matrix are shared. 
Then, two parties use garbled circuit to implement the computation
operations in DPLL, while all indexes being accessed are revealed and become
public~(Algorithm~\ref{alg:gc-dpll}). As such, 
the high overhead of accessing private inputs in the
garbled circuit is avoided. In addition, only \emph{the structure
of the backtracking-search steps} is exposed as the permutation is unknown to any single party. More precisely, the history of the
search is exposed, but it is a history of permuted Boolean literals and thus
no party knows which literal of original formula is being
accessed. It is more difficult for $\A$ and $\B$ to infer $\mathbf{F}^\B$ and
$\mathbf{F}^\A$ by observing this permuted search pattern. The steps of the
privacy-preserving SAT solver are presented as follows:


\para{Step 0: Initialization (Algorithm~\ref{alg:init})}.
The solver starts by zero-padding and aligning
$\mathbf{F}^\A$ and $\mathbf{F}^{B}$. This is because during the modeling phase,
$\A$ and $\B$ each independently introduce auxiliary Boolean variables to
transform their SAT formulas into CNF. As such, $\A$ and $\B$ each reveal to
each other the number of auxiliary Boolean variables they introduced. Then
$\mathbf{F}^{\A}$ and $\mathbf{F}^{\B}$ are both "stretched" to the same number
$n$ of rows such that the row vectors $\mathbf{f}_i^{\A}$ and $\mathbf{f}_i^{\B}$
record the occurrence and polarity of the same Boolean variable $x_i$ in
$\mathbf{F}^{\A}$ and $\mathbf{F}^{\B}$, separately. Given an auxiliary Boolean
variable introduced by $\A$, its corresponding row vector in $\mathbf{F}^{\B}$
is a zero vector. 
We denote $m^\A$ and $m^\B$ the number of clause of $F^\A$ and $F^\B$ respectively,  assume $\F^{\A}$ is of
$n$-by-$m^{\A}$, $\F^{\B}$ is of $n$-by-$m^{\B}$
and let $m=m^{\A}+m^{\B}$.

After $\mathbf{F}^\A$ and $\mathbf{F}^{B}$ are aligned, the solver uses a simple
procedure as shown in Algorithm~\ref{alg:init} to store an encrypted version of
the matrix $\{\mathbf{F}^\A \ \ \mathbf{F}^\B\}$ 
privately at $\A$ where that is the concatenation of matrix $\mat{F}^\A$ and matrix $\mat{F}^\B$. At the same time, it stores the encryption key $k_B$ used privately at $\B$.  
In this work, we explicitly for encryption use one-time pad with pseudorandom function (\textbf{PRF}) that takes the a key and inputs and outputs an almost random string of the same size as the input. We use element-wise application of $\textbf{PRF}$ when input is a matrix. That is $(\mathbf{PRF}(M))_{ij} =\mathbf{PRF}(M_{ij})$.
We refer interested readers to \cite{detailGC} for more details about \textbf{PRF} and the proof security of this encryption scheme. 


\begin{algorithm}
	\footnotesize
\SetAlgoLined
	\KwIn{$A$ and $B$ each holds its private formula $\F^\A$ and $\F^\B$}
\KwOut{$\A$ holds an encrypted version of $\{\mathbf{F}^\A \ \ \mathbf{F}^\B\}$, and 
   $\B$ holds the key $k_\B$}
	$\A$ locally a generates key $k_\A$ and an $n$-by-$m^{\A}$ random matrix $\mathbf{R}^{\A}$\;
	$\A$ sends an $n$-by-$m^{\A}$ matrix $\mathbf{\Psi}
	=\Enc{{\F^\A}}{k_\A}{\mathbf{R}^{\A}}$, 
	to $\B$\;
	$\B$ locally a generates key $k_{\B}$, an $n$-by-$m^{\A}$ random matrix
	$\mathbf{R}^{\B}_{\A}$, and an $n$-by-$m^{\B}$ random matrix $\mathbf{R}^{\B}_{\B}$\;
	$\B$ sends $\mathbf{\Phi}^{\A}=\Enc{\mathbf{\Psi}}{k_\B}{\mathbf{R^{\B}_{\A}}}$,
	$\mathbf{\Phi}^{\B} =\Enc{\F^\B}{k_B}{R^{\B}_{\B}}$,
	$\mathbf{R}^{\B}_{\A}$ and $\mathbf{R}^{\B}_{\B}$ then sends it back to $\A$\;
	$\A$ computes
	$ \Enc{ \mathbf{\Phi}^{\A}}{k_\A}{\mathbf{R}^{\A}} =
	\Enc{\F^{\A}}{k_\B}{\mathbf{R}^{B}_{A}}$\;
	$\A$ stores $\mathbf{R}^{\B}= \{
		\mathbf{R}^{B}_{A} \ \ \mathbf{R}^{B}_{B}\}$
and
	$\F_{k_\B, \mathbf{R}^{\B}} =
	\{\Enc{\F^{\A}}{k_\B}{\mathbf{R}^{B}_{A}} \ \
	\Enc{\F^{\B}}{k_\B}{\mathbf{R}^{B}_{B}}\} = \Enc{\{\F^{\A} \
	\F^{\B}\}}{k_\B}{\mathbf{R}^\B}$, and $\B$ stores $k_B$\;
\caption{Preparing for oblivious shuffling.}
\label{alg:init}
\end{algorithm}

\para{Step 1: Oblivious shuffling of private inputs (Algorithm~\ref{alg:os})}.
Given a matrix $\mathbf{V}$, the oblivious shuffling algorithm in the proposed
solver (Algorithm~\ref{alg:os}) decides a row permutation $\pi$ on $\mathbf{V}$,
and distributes the secret $(\pi, \pi(\mathbf{V})$ between $\A$ and $\B$.
Specifically, it takes $\mathbf{R}$ and $\Enc{\mathbf{V}}{k_B}{\mathbf{R}}$,
the information stored at $A$ at the end of the
preparation (Line 6 of Algorithm~\ref{alg:init}), where
$\mathbf{V}=\mathbf{F}=\{\F^\A \
\F^\B\}$ and $\mathbf{R}=\mathbf{R}^B$; and $k_\B$, the
information stored at $B$ (Line 7 of Algorithm~\ref{alg:init}), as the input. 
It then has $\A$ and $\B$ each introduce private random matrices (\ie,
$\mathbf{R}^\A$ by $\A$ and $\mathbf{R}^1$, $\mathbf{R}^2$ by $\B$) and private
permutations (\ie, $\pi_\A$ by $\A$ and $\pi_\B$ by $\B$), and achieves the goal
of oblivious shuffling by local permutations and secure evaluation of \textbf{PRF} via garbled circuits. 

The key ingredient of this shuffling is the math behind Line 8, where it can be
derived that $\mathbf{\Theta}= \pi_\A(\mathbf{V} \oplus \mathbf{\Gamma} \oplus
\mathbf{\Omega})$. With this result, in Line 11-12, $\A$ can store a piece of
secret $s_\A = \pi_\B \pi_\A(\mathbf{V}) \oplus \pi_\B(\mathbf{PRF}(k_\B,
\mathbf{R}^1))$, and $\B$ can store another piece $s_\B = \pi_\B(\mathbf{PRF}(k_\B,
\mathbf{R}^1))$. As such,  the secret can only be recovered as  $\pi=s_\A \oplus
s_\B = \pi_\B
\pi_\A$, and $\pi(\mathbf{V})=\pi_\B \pi_\A(\mathbf{V}$ when both pieces are
combined together.  During this shuffling, no input from $A$ ( \ie, $\mathbf{R}$
and $\Enc{\mathbf{V}}{k_B}{\mathbf{R}}$) or $B$ ($k_\B$) is exposed to the other
party.

In addition to permuting the SAT formula matrix $\F = \{\F^\A \ \F^\B\}$ by
row with $\pi$ and distributing this secret between $\A$ and $\B$,
The proposed SAT solver also uses Algorithms~\ref{alg:init} and~\ref{alg:os} to 
permute two row vectors $\pr$ and $\as$ with $\pi$ and
distribute each resulting secret between $\A$ and $\B$. 
These two vectors represent the branching search strategy specified by the agreement provider
$\B$. Specifically, $\pr$ represents the priorities of all Boolean
variables in the verification formula $\F$. During the DPLL search, when there is no more unit literal
in the formula, the variable that is not decided and has the highest priority will be
selected as the next variable $x_i$ and first be assigned a truth value of $\asi_i$
to simplify the formula. How $\pr$ and $\as$ are computed at $\B$ is out of
the scope of the proposed solver.~\footnote{We note that the proposed solver may be more
efficient if the search strategy is decided by $\A$ and $\B$ collectively, and
this is left for future work.}

\begin{algorithm}[t]
	\footnotesize
\SetAlgoLined
	\KwIn{$\A$ holds $\mathbf{R}$ and
	$\Enc{\mathbf{V}}{k_B}{\mathbf{R}}$; $\B$ holds $k_\B$;
}
\KwOut{
	$\A$ holds  $\pi_\B \pi_\A(\mathbf{V}) \oplus \pi_\B(\mathbf{PRF}(k_\B,
	\mathbf{R}^1))$\;
	$\B$ holds $\pi_\B(\mathbf{PRF}(k_\B, \mathbf{R}^1))$;
}
	$\A$ generates a key $k_\A$, a random matrix $\mathbf{R}^\A$, and a
	permutation $\pi_\A$\;
	$\A$ computes
	$\mathbf{\Sigma} = \pi_\A(\Enc{\Enc{\mathbf{V}}{k_B}{\mathbf{R}}}{k_\A}{\mathbf{R}^\A})$,
	$\pi_\A(\mathbf{R}^\A)$ and $\pi_\A(\mathbf{R})$\; 
	$\A$ sends
	$\Sigma$
	and $\pi_\A(\mathbf{R}^\A)$ to $B$\;
	$\B$ generates two random matrices $\mathbf{R}^1$ and $\mathbf{R}^2$\;		
	$\A$ and $\B$ construct two garbled circuits to collectively compute
	$\mathbf{\Delta}=\mathbf{PRF}(k_\B, \pi_A(\mathbf{R}))$ and 
	$\mathbf{\Gamma}=\mathbf{PRF}(k_\A, \mathbf{R}^2)$, respectively,
	without revealing the input to each other;
	and the circuits do not release the results to any party\;
	$\A$ computes $\mathbf{\Lambda}= \mathbf{PRF}(k_\A,
	\pi_A(\mathbf{R}^\A))$\;
	$\B$ computes $\mathbf{\Omega}= \mathbf{PRF}(k_\B, \mathbf{R}^1)$\;
	$\A$ and $\B$ construct a garbled circuit taking the output of the
	previous circuits to collectively compute	
	$\mathbf{\Theta} =\Sigma
	\oplus \mathbf{\Delta} \oplus \mathbf{\Gamma} \oplus \mathbf{\Lambda}
	\oplus \mathbf{\Omega}$, which equals to $\pi_\A(\mathbf{V}) \oplus
	\mathbf{\Gamma} \oplus \mathbf{\Omega}$, and the
	circuit only release $\mathbf{\Theta}$ to $\B$\;
	$\B$ generates a permutation $\pi_\B$, computes $\pi_\B(\mathbf{\Theta}) = \pi_\B \pi_\A (\mathbf{V})
	\oplus \pi_\B(\mathbf{\Gamma}) \oplus \pi_\B(\mathbf{\Omega}))$, $\pi_\B(\mathbf{R}^1)$, and
	$\pi_\B(\mathbf{R}^2)$, and sends them to $\A$\;
	$\A$ computes $\mathbf{PRF}(k_\A, \pi_\B(\mathbf{R}^2))
	=\pi_\B(\Gamma)$\;
	$\A$ stores $s_\A(\mathbf{V}) = \pi_\B(\mathbf{\Theta}) \oplus \pi_\B(\mathbf{\Omega}) = \pi_\B \pi_\A(\mathbf{V}) \oplus
	\pi_\B(\mathbf{PRF}(k_\B, \mathbf{R}^1))$\; 
	$\B$ stores $s_\B(\mathbf{V}) = \pi_\B(\mathbf{PRF}(k_\B, \mathbf{R}^1))$\;
\caption{Oblivious shuffling over private inputs.}
\label{alg:os}
\end{algorithm}

\para{Step 2: A DPLL garbled circuit with exposed, permuted search history
(Algorithm~\ref{alg:gc-dpll})}.  
After using an oblivious shuffling algorithm to permute the verification formula
$\F$ and the branching search strategy $(\pr, \as)$, the proposed SAT solver
designs a garbled circuit that encodes the DPLL algorithm with an exposed, permuted search history~(Algorithm~\ref{alg:gc-dpll}).  Specifically, Algorithm~\ref{alg:gc-dpll} uses the same backtrack-searching process and the unit literal search optimization as in the original DPLL.\footnote{The proposed
solver does not use pure literal elimination because it is rarely used by modern DPLL-based
SAT solvers.}

\begin{algorithm}[t]
	\footnotesize
	\KwIn{Secrets $\pi(\F)$, $\pi(\pr)$ and $\pi(\as)$ that are distributed stored at $\A$ and $\B$;}
	$\textbf{u} = \{\text{false}, \cdots, \text{false}\}$ \; $\textbf{d} = \{\text{false}, \cdots, \text{false}\}$\;
    $\textbf{c} = \{\text{false}, \cdots, \text{false}\}$\;
	\tcp {$\textbf{u}, \textbf{d}, \in \{\text{true}, \text{false}\}^n$, $\textbf{c} \in \{\text{true}, \text{false}\}^m$} 
	$T$ is an empty stack\;
\tcp{$T$ is the stack that is initialized to be empty}
$i = 0$ \;
\While {\text{true}}{
    \If {$i \neq 0$}{
        \textbf{OblivRes}($i,\pi(\F),\pi(\as), \textbf{c}$)\;
        $([b_s], [b_c]) = \mat{OblivCC}(\pi(\F), )$\;
        $[b_c].\text{release}(\A, \B)$\;
        $[b_s].\text{release}(\A, \B)$\;
        \eIf {$b_c$} {
            \While { $\neg T.\text{empty} $}{
              $(\pi(F), i, c, d,\pi(\as), state) = T.\text{pop}()$\;
                \If {$\text{state} == \textit{FIRST}$}{
                    \textbf{break}}
            }
            \If {$T.\text{empty}$}{
                \Return \text{false}
                }
          $[\pi(\as)_i] = [\neg \pi(\as)_i]$\;
            $T.\text{push}(\pi(\F), i, \textbf{c}, \textbf{d},\pi(\as), \text{SECOND})$\;
            \textbf{continue}
            }
        { \If{$b_s$}{
            \Return \text{true}
          }
        }
    }

    $[\text{ind}] = \textbf{OblivULS}(\pi(\F),\pi(\pr), \textbf{u},\pi(\as))$\;
    $\text{[ind].release}(\A, \B)$\;
    \eIf {$ind \neq 0$}{
        $i = \text{ind}$ \;
        $d_i = \text{true}$
        }
    {
        $[i] = \textbf{OblivBranch}(\pi(\pr), \textbf{d})$\;
        $[i].\text{release}(\A, \B)$\;
        $d_i = \text{true} $\;
        $T.\text{push}(\pi(\F), i, \textbf{c}, \textbf{d},\pi(\as), \text{FIRST})$\;
        }
    $i = 1$\
    ;
}
\caption{Garbled-Circuit DPLL}
\label{alg:gc-dpll}
\end{algorithm}

Compared with the DPLL algorithm, the key differences in
Algorithm~\ref{alg:gc-dpll} are the use of oblivious subroutines, \ie,
resolution (Line 8)~\ref{alg:or}, 
contradiction detection and checking (Line 9)~\ref{alg:oc},
unit literal search (Line 25)~\ref{alg:osul}, and branching (Line
31)~\ref{alg:ob}, which are all
encoded by garbled circuits. Specifically, the operator $[\cdot]$ is used in
Algorithm~\ref{alg:gc-dpll} to indicate that the bracketed expression is
computed using a garbled circuit and the result is not released to either $\A$
or $\B$ unless explicitly specified using the \emph{release} keyword. 





\begin{algorithm}[t]
 	\footnotesize
 \SetAlgoLined
\KwIn{$i_0, [\pi(\F)], [\pi(\as)],[\textbf{c}]$}
\KwOut{\text{None}}
\For {j = 1 to m }{
 $[b] = [\pi(\as)_{i_0}]$\;
 $[\text{cond}_0] = [\neg \pi(\F)_{i_0, j}.\OC ]$\;

 $[\text{cond}_1] = [(b \neq  \pi(\F)_{i_0, j}.P)\cdot \pi(\F)_{i_0, j}.\OC ]$\;
 $[\text{cond}_2] = [(b == \pi(\F)_{i_0, j}.P)\cdot \pi(\F)_{i_0, j}.\OC ]$\;

\For {i = 1 to n }{
 $[\text{cond}_3] = [i == i_0 ]$\;

 $[\pi(\F)_{i,j}.\OC] = [ \text{cond}_0 \cdot \pi(\F)_{i,j}.\OC +\text{cond}_2.0 + \text{cond}_1\cdot \text{cond}_3\cdot 0 + \text{cond}_1\cdot (1-\text{cond}_3)\cdot \pi(\F)_{i,j}.\OC ]$\;
}
$\textbf{c}_j = \text{cond}_2$
}
\caption{\textbf{OblivRes}: oblivious resolution.}
\label{alg:or}
\end{algorithm}

Specifically, the algorithm takes $\pi(\F)$, $\pi(\pr)$ and $\pi(\as)$, the permuted verficiation formula and search strategy from the oblivious shuffling
step, as input. These inputs are split into two secrets and distributed over $\A$ and $\B$, As such, for the simplicity of presentation without causing confusion, when these inputs show in the psuedocode, they represent the operations of combining the secrets held at $\A$ and $\B$.
For example, if $\pi(\F)$ is distributed as $s_A(\F)$ to $\A$ and $s_B(\F)$ to $\B$, where  $s_A(\F) \oplus s_B(\F)$, then $[\pi(\F)]$
represents $[s_A(\F) \oplus s_B(\F)]$.

Moreover, two parties also initialize  other three vectors: secret shared $\textbf{u}$ that represents whether the literals are unit or not; public $\textbf{d}$ that records if the values of the literals have been decided, and $\textbf{c}$ that records if the clauses have been removed (satisfied). Last but not least, a stack $T$ is  also instantiated for backtracking later.

 \begin{algorithm}[t]
 \footnotesize

 \SetAlgoLined
\KwIn{$ [\pi(\F)],[\textbf{c}]$}
\KwOut{$[b_s],[b_c]$}
$[b] = \sf false$\;
$[s] = 0 $\;
\For {j = 1 to m }{
$[z] = 0$ \;
\For {i = 1 to n }{
 $[z] = [z+ \FS_{i,j}.\OC \cdot 1]$\;
}
$[b] = [b \lor (z == 0 \land \neg \textbf{c}_j)]$ \; 
$[s] = [s + z]$
}
\Return $([s== 0], [b])$
\caption{\textbf{OblivCC}: oblivious contradiction checking.}
\label{alg:oc}
\end{algorithm}

\subsection{Privacy-Preserving SAT Solver: Analysis}\label{sec:analysis}
We conduct rigorous analysis on the correctness, privacy and
efficiency on the proposed privacy-preserving SAT solver.
The security and correctness of this solver rely on the
security and correctness of the oblivious shuffling phase (\ie,
Algorithm~\ref{alg:os}) and the way it performs
permutation. We hereby state them first here:
\begin{thm}[\textbf{Correctness and Privacy-Preserving of Oblivious
	Shuffling}] 
	\label{pisecure}
Given a pseudorandom function (\textbf{PRF}) and  input  $(\EC{\mathbf{V}}{k_\B}{R}, R)$ , Algorithm~\ref{alg:os}:
	\begin{itemize}[leftmargin=*]
    \item outputs $s_\A(\mathbf{V})$ and $s_\B(\mathbf{V})$  satisfying $s_\A(\mathbf{V}) \oplus  s_\B(\mathbf{V}) = \pi_\B \pi_\A(\mathbf{V})  $
    \item where $\pi_\A$ (respectively $\pi_\B$) is unknown to $\B$  ( respectively $\A$);
    \item keeps the $v_i$ unrevealed. In particular, given two vectors $\mat{V}$ and $\mat{V'}$, there is no probabilistic polynomial time $\A$ or $\B$  that is able to individually distinguish the algorithm is performed with  $(\EC{\mathbf{V}}{k_\B}{R}, R)$ or $(\EC{\mathbf{V'}}{k_\B}{R}, R)$ as an input.
    \end{itemize}
\end{thm}
\vspace{-0.5em}
\begin{proof}(Sketch)
For the first point on correctness, it can be shown as 
	{\footnotesize
	\[
    s_\A(\mathbf{V}) \oplus  s_\B(\mathbf{V}) = \pi_\B(\mathbf{PRF}(k_\B, \mathbf{R}^1)) \oplus \pi_\B \pi_\A(\mathbf{V})= \pi_\B \pi_\A(\mathbf{V}) 
\]}
For the last two points 
	on privacy, the basic idea is that the security of this oblivious
	shuffling relies on security of one-time pad encryption and garbled
	circuits that is proved in \cite{detailGC}. The central idea of our
	proof is reduction from the attacks against the privacy of this
	oblivious shuffling scheme to the attacks against either garbled
	circuits or the one-time pad encryption. 
\end{proof}


\vspace{-0.5em}
\para{Correctness}:   
The correctness of our SAT solver is a natural result of the basic idea behind our algorithm. Notice that permuting the literals will not change the satisfiablity of a formula, to prove the correctness of our protocol, it is sufficient to show that algorithm ~\ref{alg:gc-dpll} indeed implements DPLL. 
\begin{thm} [\textbf{Correctness}]\label{correctness} 
Given any SAT propositional formula $F$ in CNF as input, it is satisfiable if and only if algorithm  ~\ref{alg:gc-dpll} returns $true$. 
\end{thm}
\begin{proof}(Sketch)  
	The basic idea of the proof is to show that Algorithm~\ref{alg:gc-dpll}
	implements DPLL. 
	Essentially, we conduct as many \textbf{OblivULS}
	(Algorithm~\ref{alg:osul} 
	as possible until no unit literal is available (corresponding to searching unit literal and removing clauses). 
	Then, \textbf{OblivBranch} (Algorithm~\ref{alg:ob}) is called and
	current state is pushed into stack, after which \textbf{OblivRes}
	(Algorithm~\ref{alg:or})and \textbf{OblivCC} (Algorithm~\ref{alg:oc}) 
	are called to simplify $\pi(\F)$ and test validity of our guess (corresponding to the branching).
Whenever a failure ($b_c= \text{false}$) arises, we backtrack the searching tree
	until a \textit{first} guess is found by popping out the stack (corresponding
	backtracking).  On the other hand, if a success is met,  \ie\ $b_s=
	\text{true}$, the algorithm terminates and outputs true. Since each
	variable has at most 2 possible assignments, the algorithm either
	eventually returns true at some point, or it exhausts all possible
	assignments and reaches the bottom of stack $T$ then returns false. With
	this sketch, a formal proof can be easily derived. 
\end{proof}


\begin{algorithm}[t]
\footnotesize

\SetAlgoLined
\KwIn{$\pi(\F),  [\pi(\pr)], [\pi(\as)], [\textbf{u}]$}
\KwOut{$[\text{ind}]$  such that $\text{ind} \in \{1,\cdots, n\}$}
\For {$j = 1$ to $m $}
{$[b_j] = 0$\;
\For {$i = 1$ to $n$}
 { $[b_j] = [~b_j + \pi(\F)_{i,j}.\OC]$}
\For {$i = 1$ to $n$}
{$[\text{cond}] = [(b_j ==1) \land \pi(\F)_{i,j}.\OC ==1)] $\;
$[\textbf{u}_i] = [\text{cond} \cdot 1 + \neg \text{cond}\cdot \textbf{u}_i]$\;
$[\pi(\as)_i] = [ \textbf{u}_i \cdot  \pi(\F)_{i,j}.P + (1-\textbf{u}_i)\cdot \pi(\as)_i]$\;
}
}

$[\text{ind}] = [0]$\;
$[\text{pri}] = [0]$\;
\For {i = 1 to n}{
 $[\text{cond}]=[(\textbf{u}_i == 1) \cdot (\pi(\pr)_i > \text{pri})]$\;
 $[\text{ind}] = [\text{cond}\cdot i + \text{ind} \cdot(1-\text{cond})]$ \;
 $[\text{pri}] = [\text{cond} \cdot \pi(\pr)_i+ \text{pri} \cdot(1-\text{cond})]$\;
} 
\Return $[\text{ind}]$
	\caption{\textbf{OblivULS}: oblivious unit literal search.}
\label{alg:osul}
\end{algorithm}

\begin{algorithm}[t]
\footnotesize

\SetAlgoLined
\KwIn{$[\pi(\pr)], \textbf{d}$}
\KwOut{$[\text{ind}']$  such that $\text{ind} \in \{1,\cdots, n\}$}
$[\text{ind}'] = [0]$\;
$[\text{pri}] = [0]$\;
\For{i = 1 to n}{
 $[\text{cond}]=[\neg \textbf{d}_{i} \cdot (\pi(\pr)_i > \text{pri})]$ \;
 $[\text{ind'}] = [\text{cond} \cdot i+ \text{ind}' \cdot(1-\text{cond})]$\;
 $[\text{pri}] = [\text{cond} \cdot \pi(\pr)_i+ \text{ind}' \cdot(1-\text{cond})]$ \;
}
\Return $\text{ind}'$ \;
\caption{\textbf{OblivBranch}: oblivious branching.}
\label{alg:ob}
\end{algorithm}



\para{Security}. We give the following theorem on the information leakage of our
solver.

\begin{thm}[\textbf{Information Leakage}]
The algorithm \ref{alg:gc-dpll} realizes secure MPC for deciding the
	satisfiability of $F^\A\land F^\B $ with a leakage profile of: 
		(1) the number of common variables of $F^\A$ and $F^\B$, 
     (2) the number of variables of  $F^\A$ and $F^\B$, 
		(3) the size of the input formulas,
		and (4) the searching pattern.
\end{thm}


Essentially, we allow  $\A$ and $\B$ to learn searching pattern, \ie,  the
structure of the decision tree in DPLL, and the searching time.  Given a
formula, the backtrack search algorithm  organizes the search for a satisfying
assignment by maintaining a decision tree. 
By searching pattern, we mean the
structure of the traversal tree during the execution of DPLL. 
The searching pattern
essentially only shows when branching and backtracking occurs but does not
reveals what the branching literal is.  We stress that we can still show
security of the resulting system by formally proving the leakage is nothing more
than this. We accept this certain minimal amount of leakage is unavoidable for
the sake of efficiency. The formal proof of the theorem is based on simulation
paradigm of cryptography ~\cite{simulator} and omitted due to space limit.

We conclude our analysis with the following theorem on the efficiency of the
proposed solver.

\begin{thm}[\textbf{Efficiency}]\label{thm:efficiency}
Given a formula $f$, denote $T$ the number of steps in DPLL to determine if $f$ is satisfiable with the heuristics that can be expressed as two vectors $\as$ and $\pr$, about which we discuss in ~\ref{subsec: ppsat}, then the computational complexity as well as the communication complexity of GC-DPLL is $O(T\cdot mn)$.
\end{thm}
\vspace{-1em}
\begin{proof}
The proof can be directly derived from two facts (a) the overhead from garbled
	circuits is constant ~\cite{Garbledcircuits} so the computational complexity and
	communication complexity of each subroutine are both bounded by $O(mn)$.
	(b) the computation and communication for oblivious shuffling is
	$O(mn)$. The other parts of our protocol are exactly the same as DPLL,
	therefore  the total computation and communication will be bounded by
	$O(T\cdot mn)$. 
\end{proof}
\vspace{-1em}


\section{Evaluation}\label{sec:evaluation}

This section presents a prototype implementation of \system{}, and evaluate its performance via extensive experiments.

\subsection{Prototype Implementation}\label{sec:implementation}


To model the interdomain peering agreements verification problem, 
the \system{} prototype uses Batfish~\cite{batfish} and
Minesweeper~\cite{minesweeper} to parse vendor-specific BGP configurations to an
SMT formula, and then transforms the SMT formula to the configuration SAT
formula $\F^\B$ using bit blasting~\cite{malik2009boolean} and Tseytin
transformation~\cite{tseitin1983complexity}. The prototype also uses a similar
transformation procedure to transform the SMT formula of peering agreements to
the agreement SAT formula $\F^\A$.

The privacy-preserving SAT solver of this prototype is implemented using OpenSSL~\cite{viega2002network} and Obliv-C~\cite{zahur2015obliv}, a GCC wrapper to embed secure computation protocols inside regular C programs. The prototype uses AES-128 as the choice of PRF (pseudorandom function). One implementation optimization in the prototype is that, during the oblivious shuffling phase, instead of using component-wise operations on each element in the matrix, we divide elements in the same row of matrix into blocks of 128 bits, and use block-wise operations. Because each element in a SAT formula matrix has only 2 bits, this optimization improves the efficiency (in terms of computation, storage and communication) of oblivious shuffling by approximately 128/2=64 times. Overall, the \system{} prototype has $\sim$2700 lines of C and Obliv-C code.

\subsection{Experiment Methodology}
The goal is to examine the efficiency of \system{} on interdomain peering
agreement verification, and evaluate its overhead. To this end, two sets of
experiments are designed. The first set is the functionality check experiment.
Specifically, we setup the example abstract topology in
Figure~\ref{fig:motivation-example} as the evaluation scenario, where the two BGP
routers of the agreement provider $B$ are configured using the Cisco IOS
configuration language, and the agreement consumer specifies five representative
peering agreements~\cite{com-guide, vperf} (including the three introduced in
Section~\ref{sec:background}). 
The second set of experiment uses a standard SAT datasets (\ie, the SATLIB 
datasets~\cite{satlib}) to evaluate the overhead
of the privacy-preserving SAT solver more extensively.
In both sets of experiments, the agreement consumer and provider each runs its \system{} process on an Ubuntu server with Intel(R) Xeon(R) Platinum 8168 CPU and 376GB RAM.



\subsection{Results}
\para{Functionality check}: 
In the first experiment, five representative peering agreements 
are specified. For each agreement, we
create different interdomain agreement verification problems by using three
different sets of correct BGP configurations of $B$. As such, a total of 15
experiments are executed.  In all 15 experiments, \system{} successfully
determines that the problem of $\F^\A \wedge \F^\B$ is unsatisfiable, indicating
that the agreement-to-verify is implemented correctly at $B$. For each
agreement, we then deliberately change one correct BGP configuration to be
incorrect, and executes \system{} again. Results show that in all these 5
experiments, \system{} successfully finds a satisfiable solution to the problem
of $\F^\A \wedge \F^\B$, indicating that these agreements are not correctly
implemented. With this functionality check result, we conclude that \system{}
provides accurate results for interdomain peering agreement verification.



\para{Overhead study via SAT benchmark simulation}: 
In the second experiment, we select an open SAT dataset~\cite{satlib} whose input SAT
formulas have 50 variables and 100 clauses. In the simulation, we vary the
number of variables ($n$) to be from 5 to 50, with a step size of 5, and vary
the number of clauses ($m$) to be from 10 to 100, with a step size of 10. For
each combination, 30 SAT formulas are generated from the dataset. In each
experiment, each formula is approximately separated in half and held by $A$ and $B$, respectively. 

The first metric we study is \textit{oblivious shuffling delay}, 
\ie, the time consumption of the oblivious shuffling phase of the proposed
solver. Figure~\ref{fig:eval:time-precompute-m} shows that (1) the worst case of
oblivious shuffling delay is less than 10 seconds; and (2) the delay increases
sharply when $m$ reaches 64. This sharp increase is the result of the
implementation optimization, where the oblivious shuffling operations in the
same row are block-wise with 128/2=64 elements, not component-wise.
Figure~\ref{fig:eval:time-precompute-n} shows that the oblivious shuffling delay
increases linearly as the number of variables increases. The separation of two
lines in this figure is again the result of our block-wise shuffling operation.

Second, we study the \textit{verification delay}, which is the sum of the
oblivious shuffling delay and the delay of the GC-DPLL algorithm.
Figure~\ref{fig:eval:time-obliv-total} shows that the oblivious shuffling delay
takes about 50\% of the verification delay, and that the impact of $m$ on the
verification delay increases as $n$ increases.

Third, we study the size of transmitted data during the verification. 
Figure~\ref{fig:eval:trans-obliv-total} shows that the size of total transmitted data is
up to 3 GB. Although this seems like large, in modern networks where the interdomain bandwidth is
often of hundreds of GBs, this data transmission overhead is negligible.

Fourth, we compare the performance of the GC-DPLL algorithm with the original
DPLL algorithm. Figure~\ref{fig:eval:cmp-log-time} compares their time
consumption. We observe that the GC-DPLL algorithm increases the verification
delay of the DPLL algorithm by three orders of magnitude. Considering that the
absolute verification delay of GC-DPLL is of tens of seconds, this increase is acceptable.
With these overhead study, we conclude that \system{} is an efficient solution for
privacy-preserving interdomain agreement verification.


\begin{figure}[t]
    \centering
    \begin{subfigure}[t]{.43\linewidth}
        \centering\includegraphics[width=\linewidth]{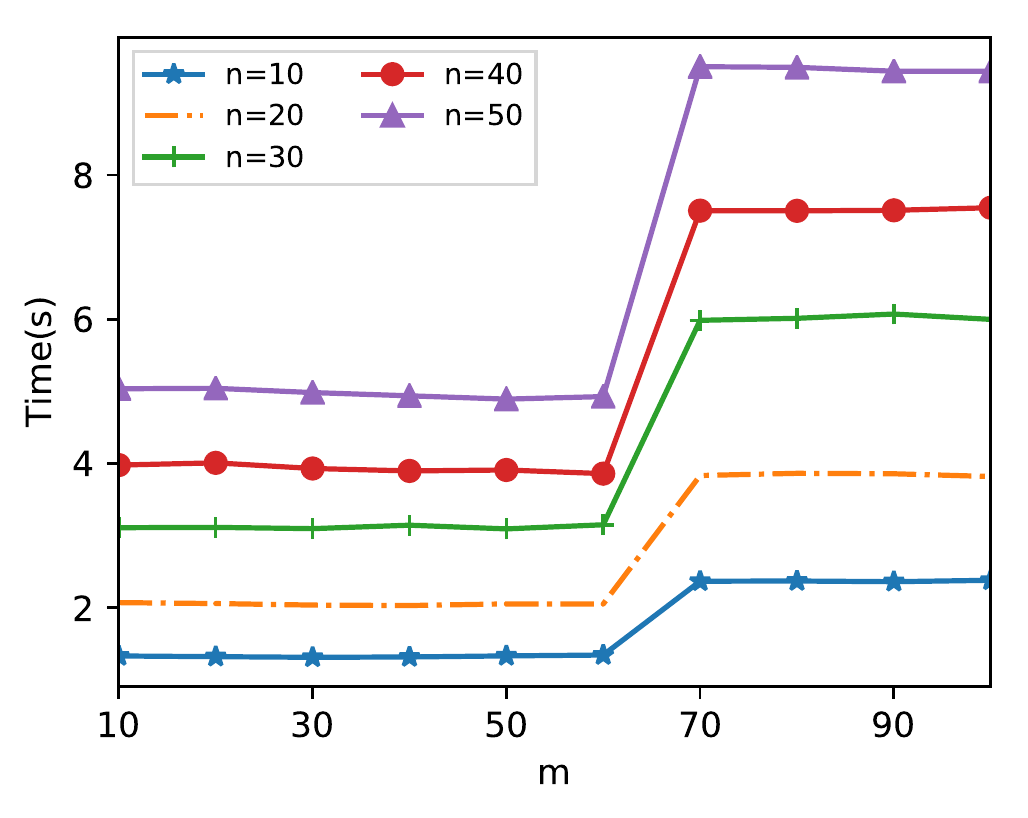}
    \caption{\label{fig:eval:time-precompute-m} \scriptsize Varying number of
	    clauses $m$.}
    \end{subfigure}
    \begin{subfigure}[t]{.43\linewidth}
        \centering\includegraphics[width=\linewidth]{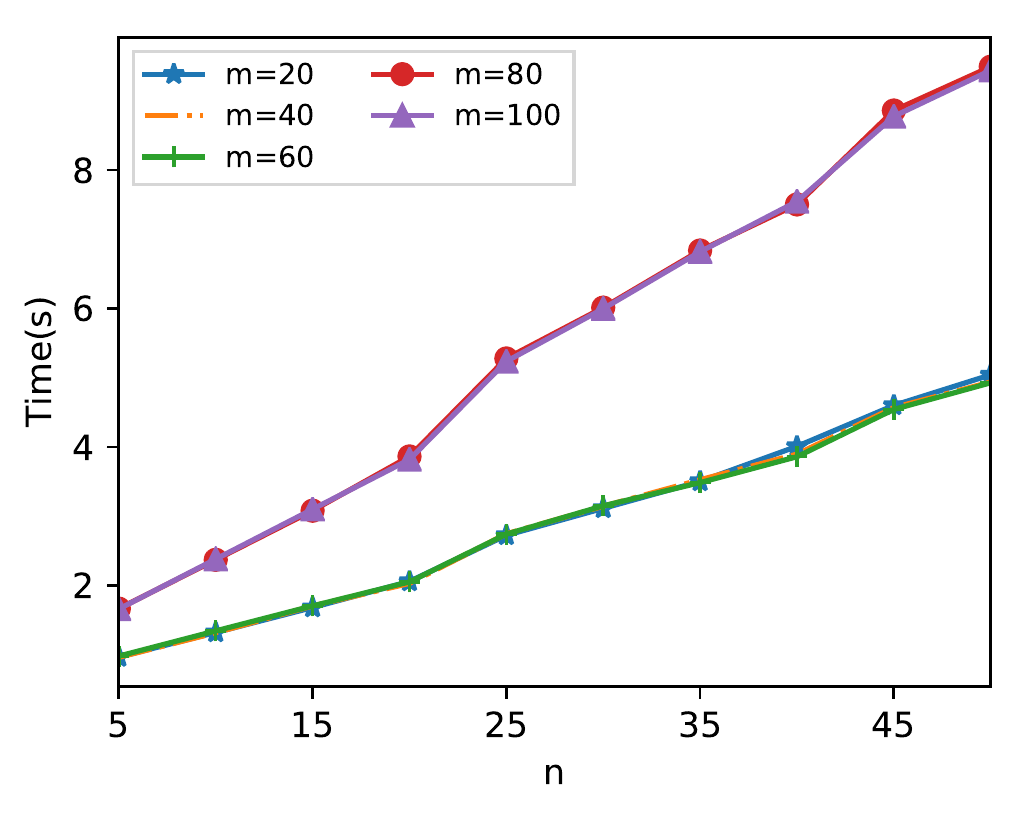}
    \caption{\label{fig:eval:time-precompute-n} \scriptsize Varying number of
	    variables $n$.}
    \end{subfigure}
	\caption{Oblivious shuffling delay of \system{}.}
    \label{fig:eval:time-precompute}
\end{figure}


\begin{figure}[t]
    \centering
    \begin{subfigure}[t]{.43\linewidth}
        \centering\includegraphics[width=\linewidth]{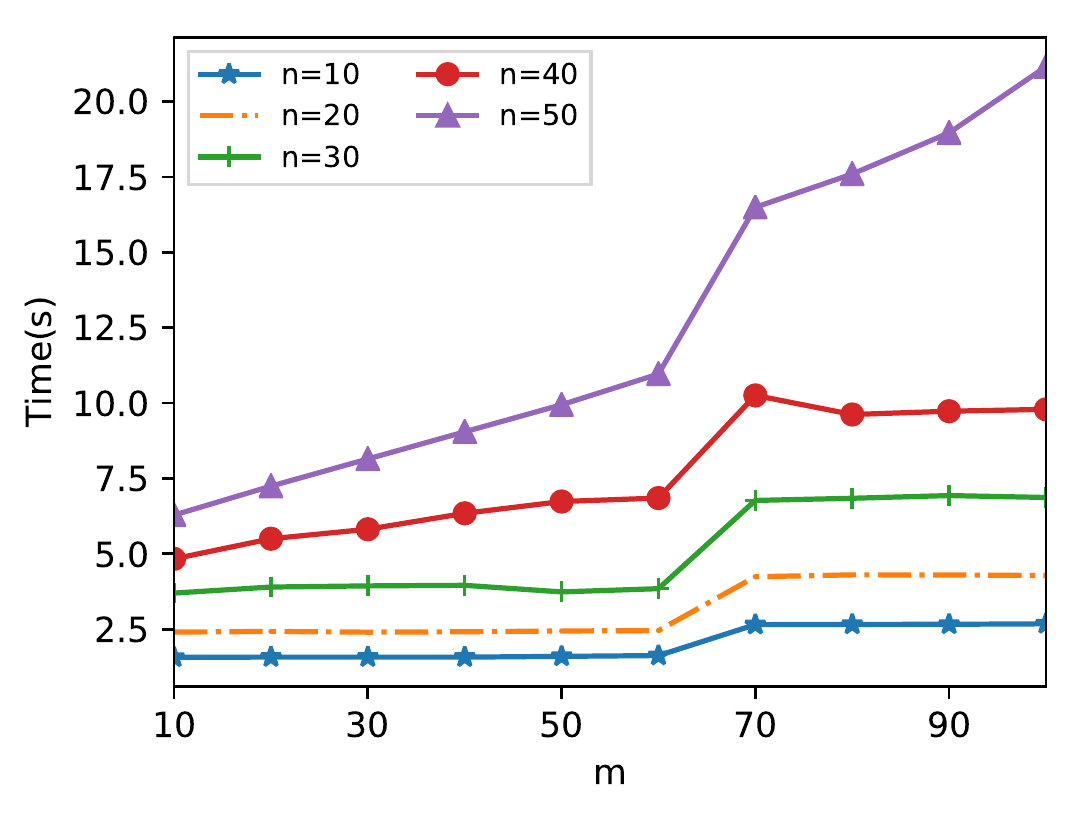}
    \caption{\label{fig:eval:time-obliv-total-m} \scriptsize Varying number of
	    clauses $m$.}
    \end{subfigure}
    \begin{subfigure}[t]{.43\linewidth}
        \centering\includegraphics[width=\linewidth]{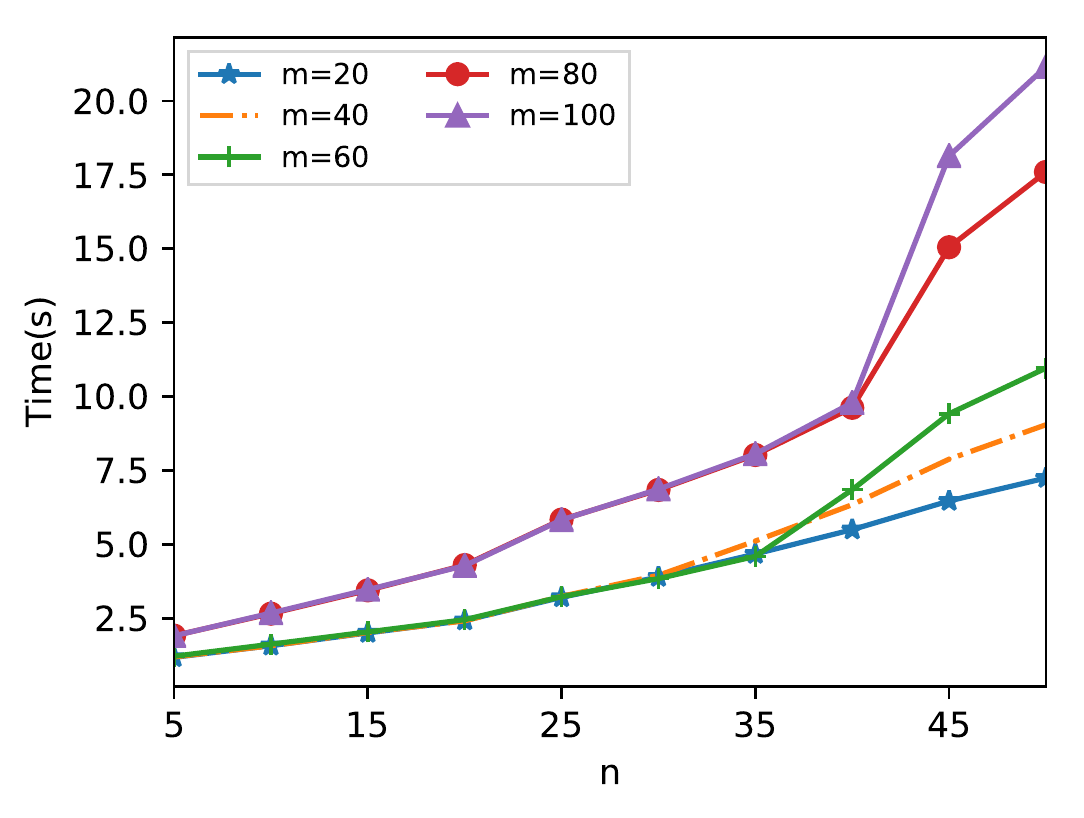}
    \caption{\label{fig:eval:time-obliv-total-n} \scriptsize Varying number of
	    variables $n$.}
    \end{subfigure}
	\caption{Verification delay of \system{}.}
    \label{fig:eval:time-obliv-total}
\end{figure}


\begin{figure}[t]
    \centering
    \begin{subfigure}[t]{.43\linewidth}
        \centering\includegraphics[width=\linewidth]{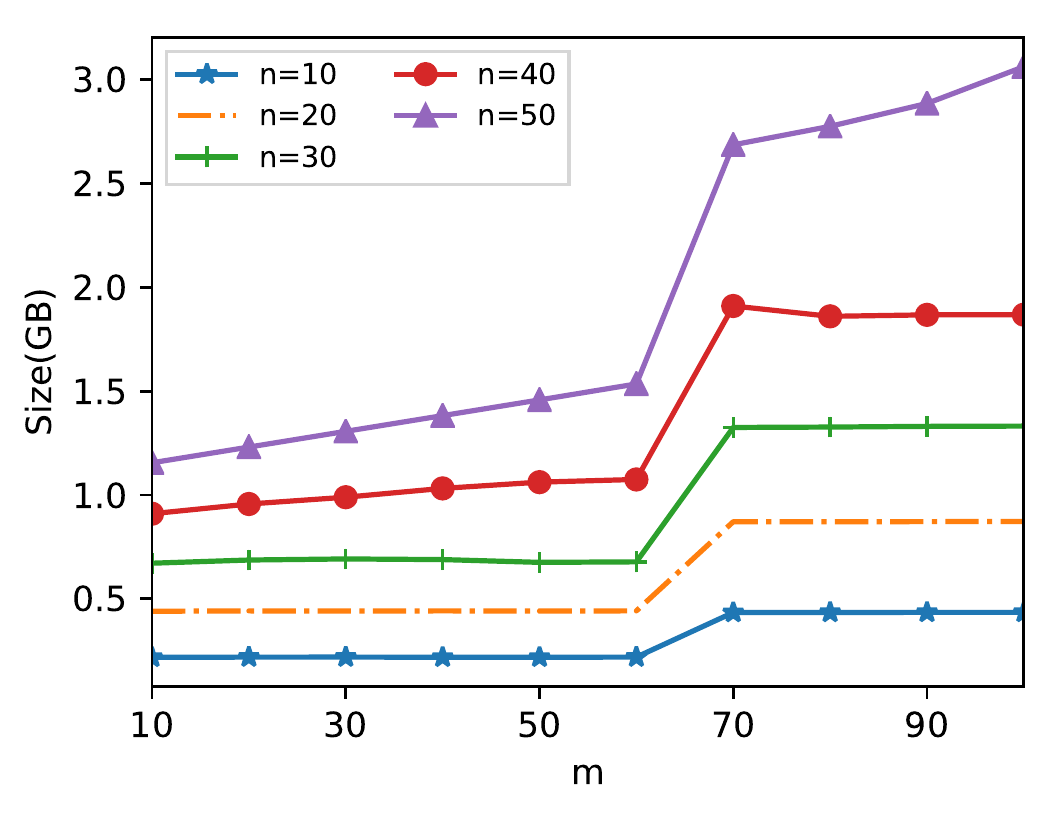}
    \caption{\label{fig:eval:trans-obliv-total-m} \scriptsize Varying number of
	    clauses $m$.}
    \end{subfigure}
    \begin{subfigure}[t]{.43\linewidth}
        \centering\includegraphics[width=\linewidth]{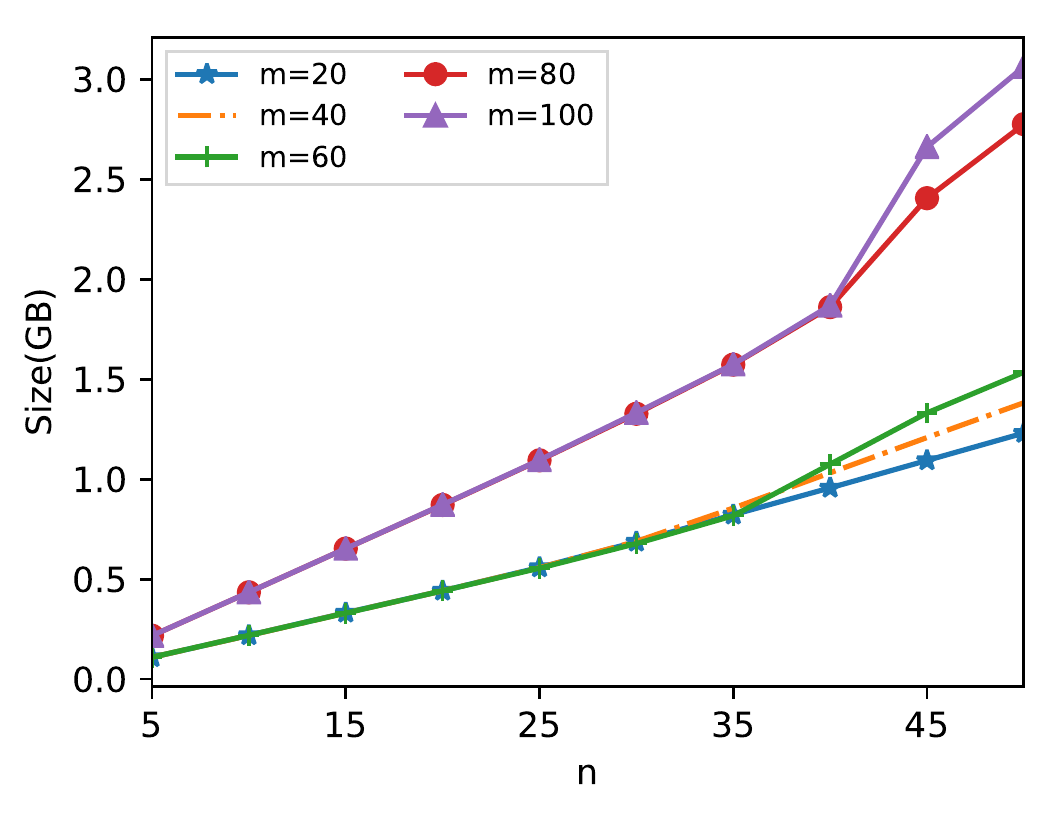}
    \caption{\label{fig:eval:trans-obliv-total-n} \scriptsize Varying number of
	    variables $n$.}
    \end{subfigure}
	\caption{Data transmission overhead of \system{}: total message size.}
    \label{fig:eval:trans-obliv-total}
\end{figure}

\begin{figure}[t]
    \centering
    \begin{subfigure}[t]{.43\linewidth}
        \centering\includegraphics[width=\linewidth]{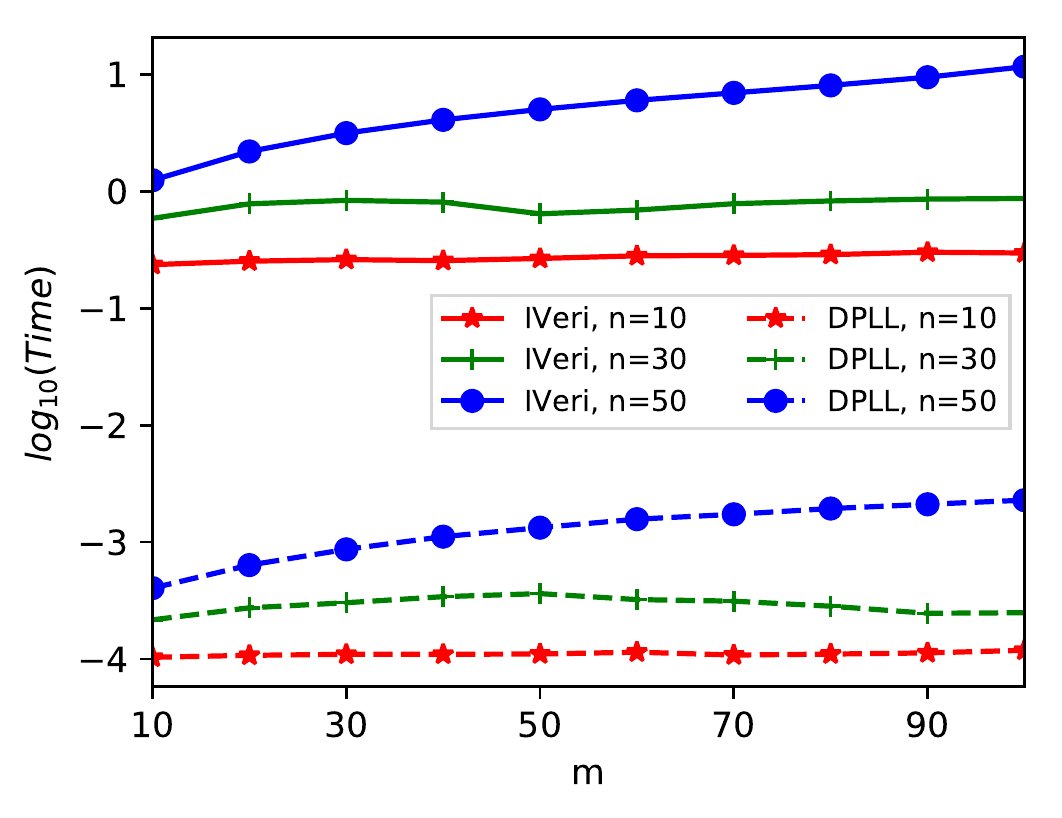}
    \caption{\label{fig:eval:cmp-log-time-m} \scriptsize Varying number of
	   clauses $m$}
    \end{subfigure}
    \begin{subfigure}[t]{.43\linewidth}
        \centering\includegraphics[width=\linewidth]{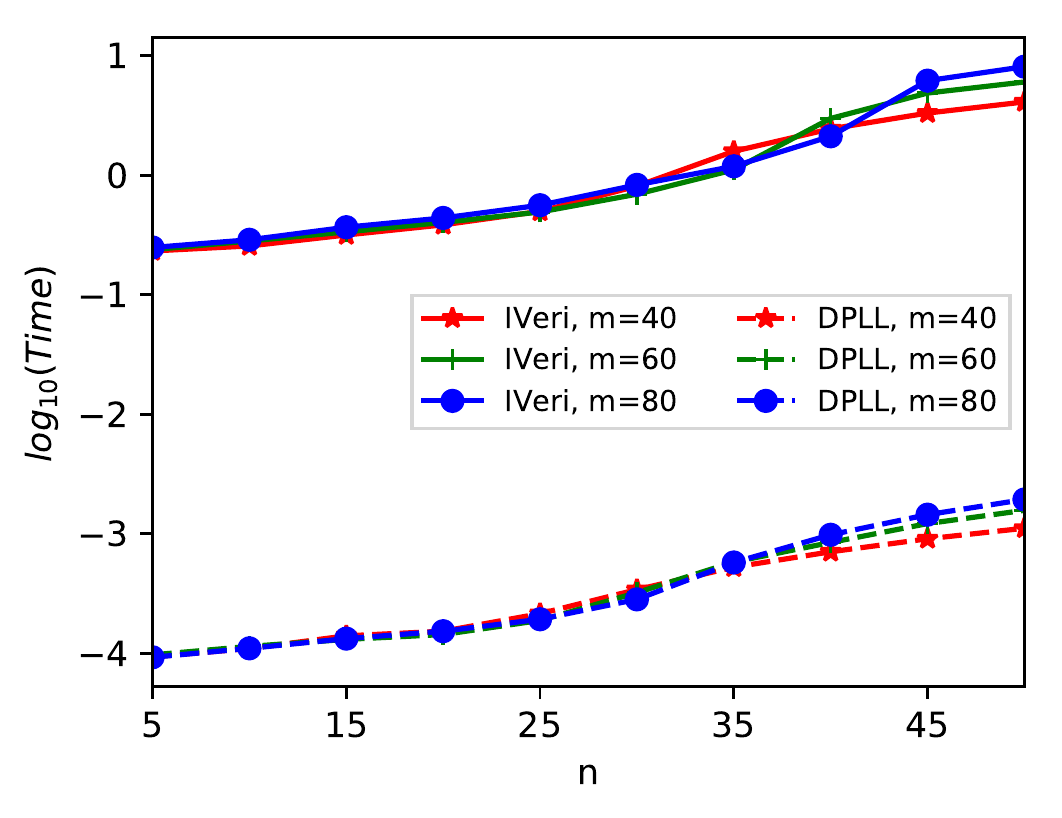}
    \caption{\label{fig:eval:cmp-log-time-n} \scriptsize Varying number of
	    variables $n$.}
    \end{subfigure}
    \caption{Computation delay: GC-DPLL vs.  DPLL.}
    \label{fig:eval:cmp-log-time}
\end{figure}


\vspace{-0.5em}
\section{Related Work}\label{sec:related}

\para{Network verification}.  
Network verification tools are roughly categorized in two classes. First,
network configuration verification tools (\eg, \cite{minesweeper, compression, era}) focus on finding errors in the configurations of
routing protocols (\ie, the control plane of network) within the same network by
building and analyzing a symbolic network model (\eg, an SMT model). \system{}
heavily leverages this approach to model BGP configurations, 
but goes beyond to build new models for
interdomain peering agreements.
Second, data plane verification tools (\eg, \cite{hsa, 
atomic-predicate, veriflow}) focus
on verifying that the data plane of a network (\eg, the
forwarding tables and access control lists) satisfies certain properties (\eg,
loop-freeness).  Although tools such as Looking
Glass~\cite{khan2013level} and RouteView~\cite{routeviewsuniversity} expose certain BGP data plane
configurations (\eg, selected routes) of participating ASes, such incomplete
information makes it non-trivial to extend data plane verification
to interdomain networks. 

\para{Verification in interdomain networks}. 
There are some efforts in interdomain network verification~\cite{vperf, fsr,
bagpipe, sidr, netreview, sam-icnp-collaborative, netquery}. FSR~\cite{fsr},
Bagpipe~\cite{bagpipe} and biNode~\cite{9155235}
focus on verifying the basic property (\ie, stability and reachability)
of BGP in interdomain networks and require the complete exposure of the BGP
configurations of all ASes. 
SIDR~\cite{sidr} focuses on
verifying the safety property of data plane configurations of multiple
SDXes, and sacrifices accuracy for privacy and safety.
NetReview~\cite{netreview} and VPerf~\cite{vperf} focus on the errors in BGP
configurations in BGP peers and are work most related to \system{}. However,
they operate on a per route announcement basis, \ie, it can
only find errors after a route announcement is sent. To the best of our
knowledge, \system{} is
the first to verify the agreement implementation in BGP configurations
proactively while preserving the privacy of ASes.


\para{SAT solver}. 
SAT solver is an active research area with a rich
literature~\cite{malik2009boolean, gomes2008satisfiability, dpll,
sorensson2005minisat, qin2014structure}. 
\cite{qin2014structure} studies how
to disguise a SAT formula held by one party from the other party who holds an
SAT solver
However, to the best of our knowledge, we are the first SAT solver 
to decide the satisfiability 
of the conjunction of two private SAT formulas held by two parties.

\para{SMPC in networks}.
SMPC-based routing systems have been proposed 
to let ASes collectively compute interdomain routes for better
network performance
~\cite{chen2018sdn,
asharov2017privacy}.  They are orthogonal to this paper, as \system{} focuses on
verifying interdomain peering agreements instead of designing new routing
systems. 


\vspace{-2mm}
\section{Conclusion}\label{sec:conclusion}
We investigate the important problem of interdomain
peering agreement verification. We identify privacy as the fundamental
challenge, and design \system{}, the first privacy-preserving interdomain
agreement verification system, whose core is 
a novel privacy-serving SAT solver.
Extensive evaluation on a \system{} prototype demonstrates its efficiency and efficacy.


\bibliographystyle{IEEEtran}
\bibliography{all}
\end{document}